\documentclass[submission,copyright,creativecommons]{eptcs}
\providecommand{\event}{SYNT 2017} 

\usepackage{macros}
\usepackage{amssymb, amsmath}
\usepackage{ntheorem}
\usepackage{wrapfig}
\usepackage{framed}
\usepackage{paralist}
\usepackage{pifont}
\usepackage{hyperref,breakurl}
\usepackage{caption}
\usepackage{subcaption}
\usepackage{balance}
\usepackage{caption,subcaption}
\usepackage{float}
\usepackage{xparse,graphics}

\newtheorem{definition}{Definition}

\newtheorem{exam}{Theorem}
\newtheorem{example}[exam]{Example}
\newtheorem{theo}{Theorem}
\newtheorem{theorem}[theo]{Theorem}

\newtheorem{lemm}{Lemma}
\newtheorem{lemma}[lemm]{Lemma}
\newtheorem*{proo}{Proof}
\newtheorem*{proof}[proo]{Proof}

\title{A Class of Control Certificates to Ensure Reach-While-Stay for Switched Systems}
\author{
Hadi Ravanbakhsh and Sriram Sankaranarayanan
\institute{Department of Computer Science\\
University of Colorado, Boulder\\
Boulder, CO, USA}
\email{firsname.lastname@colorado.edu}
}
\def\titlerunning{A Class of Control Certificates to Ensure Reach-While-Stay for Switched Systems}
\def\authorrunning{H. Ravanbakhsh \& S. Sankaranarayanan}
\begin{document}
\maketitle

\begin{abstract}
In this article, we consider the problem of synthesizing switching controllers
for temporal properties through the composition of simple
primitive reach-while-stay (RWS) properties. Reach-while-stay
properties specify that the system states starting from
an initial set $I$, must reach a goal (target) set $G$ in finite time,
while remaining inside a safe set $S$.   Our approach synthesizes
switched controllers that select between finitely many
modes to satisfy the given RWS specification. To do so, we
consider control certificates, which are Lyapunov-like functions
that represent control strategies to achieve the desired
specification. However, for RWS problems, a control Lyapunov-like
function is often hard to synthesize in a simple polynomial form. Therefore, we
combine control barrier and Lyapunov functions
with an additional compatibility condition between them. Using this approach,
the controller synthesis problem reduces to one of solving
quantified nonlinear constrained problems that are handled using
a combination of SMT solvers.  The synthesis of controllers is demonstrated
through a set of interesting numerical examples drawn from the related
work, and compared with the state-of-the-art tool SCOTS. Our evaluation
suggests that our approach is computationally feasible, and adds to
the growing body of formal approaches to controller synthesis.
\end{abstract}

\section{Introduction}\label{sec:intro}

The problem of synthesizing switching controllers for reach-while-stay
(RWS) specifications is examined in this article. RWS properties
are an important class, since we may decompose more complex temporal
specifications into a sequence of RWS specifications~\cite{kloetzer2008fully}. The
plant model is a switched system that consists of finitely many
(controllable) modes, and the dynamics for each mode are specified
using ODEs. Furthermore, we consider nonlinear ODEs for each mode,
including rational, trigonometric, and exponential functions. The goal
of the controller is to switch between the appropriate modes, so that
the resulting closed loop traces satisfy the specification.

RWS properties specify that a goal set $G$ must be reached by all
behaviors of the closed-loop system while staying inside a safe set
$S$.  Specifically, the state of the system is assumed to be
initialized to any state in the set $S$. RWS properties include safety
properties (stay inside a safe set $S$), reachability properties (reach
a goal set $G$), and ``control-to-facet''
problems~\cite{habets2004control,helwa2013monotonic}.

The controller synthesis is addressed in two phases: (a) formulating a
\emph{control certificate} whose existence guarantees the existence of a
non-Zeno switching control law for the given RWS specification, and
(b) solving for a certificate of a particular form as a feasibility
problem. The control certificates are control
Lyapunov-like functions which represent a strategy for the controller
to satisfy the specifications. Additionally, this strategy can be
effectively implemented as a feedback law using a controller that
respects min dwell time constraints. In the second phase, a
counterexample guided inductive synthesis (CEGIS) framework~\cite{solar2008program}, 
--- an approach that uses SMT solvers at its core --- is used to
discover such control certificates. However, this procedure is used
off the shelf, building upon the previous work of Ravanbakhsh et
al.~\cite{Ravanbakhsh-Others/2015/Counter-LMI}. This procedure uses a specialized solver for finding a certificate
of a given parametric form that handles quantified formulas by
alternating between a series of quantifier free formulas using
existing SMT
solvers~\cite{DeMoura+Bjorner/08/Z3,DBLP:conf/cade/GaoKC13}.


The contributions of the paper are as follows:
first, we show that a straightforward formulation of the control
certificate for the RWS problem yields an exponential number of
conditions, and hence can be computationally infeasible. Next, we
introduce a class of control certificates which (i) has a concise logical
structure that makes the problem of discovering the certificates
computationally feasible; and (ii) we show that such certificates
yield corresponding switching strategies with a min-dwell time
property unlike the conventional control certificates.
Next, we extend our approach to  the initialized RWS (IRWS) property that
additionally restricts the set of initial conditions of the system
using a class of ``control zero-ing'' barrier functions~\cite{wieland2007constructive,xu2015robustness} . Also, a
suitable formulation for these functions is provided within our
framework.
Finally, we provide numerical examples to demonstrate the effectiveness of the
method, including comparisons with recently developed state-of-art
automatic control synthesis tool SCOTS~\cite{rungger2016scots}.

\subsection{Related Work}

The broader area of temporal logic synthesis
seeks to synthesize formally guaranteed controllers
from the given plant model and specifications. 
The dominant approach is to build a discrete abstraction of the given
plant that is related to the original
system~\cite{wongpiromsarn2011tulip,liu2013synthesis,mazo2010pessoa,rungger2016scots,mouelhi2013cosyma}. Once
a suitable abstraction is found, these approaches use a
systematic temporal logic-based controller design approach over the
abstraction~\cite{thomas2002automata}. The properties of interest in
these systems include the full linear temporal logic (LTL) and an
efficiently synthesizable subset such as
GR(1)~\cite{wongpiromsarn2011tulip,liu2013synthesis}. These approaches
differ in how the abstraction can be constructed in a guaranteed
manner.  One class of approaches works by fixing a time step, gridding
the state-space, and simulating one point per
cell~\cite{mouelhi2013cosyma,mazo2010pessoa,rungger2016scots,zamani2012symbolic,tazaki2012discrete}. The
resulting abstraction, however, is not always approximately bisimilar to the original system. Nevertheless, conditions such as open loop
incremental stability of the plant can be used to obtain bisimilarity~\cite{girard2010approximately}.
Alternatively, the abstraction
can be built without time discretization~\cite{ozay2013computing,liu2013synthesis} by considering infeasible
transitions. And furthermore, the abstraction can be iteratively refined through a counter-example refinement
scheme~\cite{nilssonincremental}. Our
work here does not \emph{directly} focus on building abstractions. Rather, our focus
is on deductive approaches for a narrow class of temporal logic
properties namely RWS properties. Using our approach, control systems
for richer properties can be built from solving a series of RWS
problems.

Our approach is closely related to work of Habets et al.~\cite{habets2006reachability} and
Kloetzer et al.~\cite{kloetzer2008fully}. In these methods, an
abstraction is obtained by solving local control-to-facet problems
instead of reachability analysis. However, continuous feedback is
synthesized for each control-to-facet problem. The key difference in
this paper is that the control-to-facet problems themselves are solved
using switching. Furthermore, we consider initialized problems, where
the initial states are also restricted to belong to a set. We find
that IRWS problems can often be realized through a controller even
when the corresponding RWS problem (for which the initial condition is not
restricted) cannot be synthesized.

Another related class of solutions is based on synthesizing 
``a deductive proof of correctness" simultaneously with ``a control strategy". The
goal of these approaches also consists of finding a control
certificate, which yields a (control) strategy to guarantee the
property.  This typically takes the form of a control Lyapunov-like
function. The idea of control Lyapunov functions goes back to
Artstein~\cite{artstein1983stabilization} and
Sontag~\cite{sontag1989universal}. The problem of \emph{discovering} a
control Lyapunov function is usually formulated using bilinear matrix
inequalities (BMI) ~\cite{tan2004searching}. Also, instead of solving
such NP-hard problems, usually alternating optimization (V-K iteration
or policy iteration) is
used to conservatively find a solution
~\cite{tan2004searching,ghaoui1994synthesis}. 

Wongpiromsarn et al.~\cite{wongpiromsarn2016automata} discuss \emph{verification} of temporal logic properties using barrier certificates.
For synthesis, Xu et al.~\cite{xu2015robustness} discuss
conditions for the so-called ``control zeroing'' barrier functions for safety
and their properties.
They also, consider their combination with control Lyapunov functions. 
In this article, we provide an alternative condition that is
based on ``exponential condition'' barrier functions~\cite{kong2013exponential} and
enforcing a compatibility condition between the control actions
suggested by the control barrier and control Lyapunov functions. Also, Dimitrova
et. al. ~\cite{dimitrova2014deductive} have shown that control
certificates can be extended to address more complicated
specifications i.e. parity games. While these results show that constraint solving
based methods can be applied on more complicated specification, no method
 of finding such certificates is provided. 
 
The use of SMT solvers in control synthesis has also been
well-studied. Taly et. al
~\cite{taly2011synthesizing,taly2010switching} use a constraint solving
approach to find control certificates for reachability and safety. They adapt
a technique known as Counter-Example Guided Inductive Synthesis
(CEGIS), originally proposed for program synthesis~\cite{solar2008program},
to solve
the control problems using a combination of an SMT solver
 with numerical simulations. Ravanbakhsh et
 al.~\cite{Ravanbakhsh-Others/2015/Counter-LMI} 
propose  a combination of SMT and SDP
solvers for finding control certificates.  However, their method is
only applicable to stability or simple reachability properties, 
involving the use of a single
Lyapunov function.  In a subsequent paper, their approach is
extended to handle disturbance inputs~\cite{Ravanbakhsh-Others/2016/EMSOFT}. The use of SMT solvers to
solve for Lyapunov-like functions is used in our paper as
well. However, this paper focuses on defining a more tractable class of control
certificates for RWS problems. Furthermore, we show how these problems
can be composed for more complex temporal objectives. In particular,
our use of the CEGIS procedure is not a contribution of this paper.
Furthermore, in order to handle nonlinear systems and
also to guarantee numerical soundness of these solvers, we use the
nonlinear SMT solver 
dReal~\cite{DBLP:conf/cade/GaoKC13}. 

Huang et.  al.~\cite{huang2015controller} also propose control
certificates to solve the RWS problem for piecewise affine systems, 
using SMT solvers. Their
approach uses piecewise constant functions as control certificates and
partitions the state space into small enough cells in order to define
such functions. By using this technique, any function can be
approximated, which makes the method relatively complete. 

As mentioned earlier, past work by Habets et al. and Klutzier et
al.~\cite{habets2006reachability,kloetzer2008fully} 
build a finite abstraction by repeatedly solving control-to-facet
problems.  These problems seek to find a feedback law inside a
polytope $P$ that guarantees all the resulting trajectories
exit $P$ through a specific facet $F$ of $P$.
Habets et. al.~\cite{habets2004control} show  necessary and
sufficient conditions for the existence of a control strategy for the
control-to-facet problem on simplices. This condition is
sufficient but not necessary for polytopes. They extract a unique certificate from each
problem instance and check whether the condition holds for the
certificate. Subsequently, Roszak et al.~\cite{roszak2006necessary} and
Helwa et al.~\cite{helwa2013monotonic} extend this 
approach and solve reachability to a set of facets by introducing flow 
condition, which combined with invariant condition serves as a control 
certificate similar to those used in this paper. From the published
results,
these methods are more
efficient, but are only applicable to affine systems over polytopes. In contrast, the dynamics
in this article can be non-linear involving rational, trigonometric,
and exponential functions. In this article, we demonstrate that our method
can be used to solve such problems and it can be integrated into other
methods which build an abstraction for the system.

\section{Background}\label{sec:background}
\subsection{Notation}

Given a function $f(t)$, let $f^+(t)$ ($f^-(t)$) be the right (left) limit of $f$ at
$t$, and $\dot{f}(t)$ represent the \emph{right derivative} of $f$ at
time $t$.  For a set $S \subseteq \reals^n$, $\partial S$ and $int(S)$ are its boundary and interior, respectively.

\begin{definition}\textit{(Nondegenerate Basic Semialgebraic Set):}\label{Def:basic-semialgebraic-set}

A \emph{nondegenerate basic semialgebraic set} $K$ is a nonempty set
defined by a conjunction polynomial inequalities:
\[ K :\ \{\vx\ |\ p_{K,1}(\vx) \leq 0\ \land\ \cdots\ \land p_{K,i}(\vx) \leq 0\} \,,\]  
where $\vx \in \reals^n$. For
each $j \in [1,i]$, we define \[ H_{K,j} =
\{\vx\ |\ \vx \in K\ \land\ p_{K,j}(\vx) = 0\} \,.\]
\end{definition}

It is required that (a) each $H_{K,j}$ is nonempty, (b) the boundary $\partial K$ and the interior $int(K)$ are given by $ \bigvee_{j=1}^i H_{K,j}$ and $\bigwedge_{j=1}^i p_{K,j}(\vx) < 0$, respectively, and (c) the interior is nonempty. We use ``basic semialgebraic" and ``nondegenerate basic semialgebraic" interchangeably. 


\subsection{Switched Systems}\label{sec:switched-systems}

We consider continuous-time switched system plants, controlled by a
memoryless controller that provides continuous-time switching feedback. The state
of the plant $\Plant$ is defined by $n$ continuous variables $\vx$ in
a state space $X \subseteq \reals^n$, along with a finite set of modes
$Q=\{q_1,\ldots,q_m\}$. The trace of the system $(q(t), \vx(t))$ maps time to mode
$q(.):\reals^+ \rightarrow Q$, and state $\vx(.):\reals^+ \rightarrow
X$.
The mode $q \in Q$ is controlled by an
external switching input $q(t)$. The state of the plant inside each
mode evolves according to (time invariant) dynamics:
\begin{equation}\label{Eq:plant-dynamics}
	\dot{\vx}(t) = f_{q(t)}(\vx(t)) \,,
\end{equation}
wherein $f_q:X \rightarrow \reals^n$ is a Lipschitz continuous
function over $X$, describing the vector field of the plant for mode
$q$.  

The controller $\Ctrl$ is defined as a function $\ctrl: Q \times X \rightarrow Q$, which
given the current mode and state of the plant, decides the mode of the
plant at the next time instant. Formally:
\begin{equation}\label{Eq:ctrl-function}
	q^+(t) = \ctrl(q(t), \vx(t)) .
\end{equation}

The closed loop $\tupleof{\Plant,\Ctrl}$ produces traces $(q(t),
\vx(t))$ defined jointly by equations~\eqref{Eq:plant-dynamics} and
~\eqref{Eq:ctrl-function}.  However, care must be taken to avoid
\emph{Zenoness}, wherein the controller can switch infinitely often in
a finite time interval. Such controllers are physically
unrealizable. Therefore, we will additionally ensure that the $\ctrl$
function satisfies a \emph{minimum dwell time} requirement that
guarantees a minimum time $\delta > 0$ between mode switches.

\begin{definition}\textit{(Minimum Dwell Time):}
 A controller $\Ctrl$ has a minimum dwell time $\delta > 0$ with
 respect to a plant $\Plant$ iff for all traces
 and for all switch times $T$ ($q(T) \neq q^+(T)$,  the controller does not
 switch during the times $t \in [T, T+\delta)$: i.e, $\ctrl(q(t), \vx(t))
   = q^+(T)$ for all  $t \in [T, T+\delta)$.
       \end{definition}

Once the function $\ctrl$ is defined with a minimum dwell time
guarantee, given initial mode ($q(0)$), and initial state ($\vx(0)$),
a unique trace is defined for the system.



\paragraph{Specifications:}
Generally, specifications describe desired sequences of plant states
$\vx(t)$ over time $t \geq 0$ that we wish to control for.
In this paper, we focus on \emph{reach-while-stay} (RWS) specifications
involving three sets: \emph{initial set} $I \subseteq X$, \emph{safe set} $S \subseteq X$
and \emph{goal set} $G \subseteq X$.
\begin{definition}\textit{(Initialized Reach-While-Stay (RWS) Specification):}
  A trace $\vx(t)$ for $t \in [0,\infty)$ satisfies a reach-while-stay (RWS) specification w.r.t
    sets $\tupleof{I,S,G}$ iff whenever $\vx(0) \in I$, there exists a time $T \geq 0$
    s.t. for all $t \in [0,T)$, $\vx(t) \in S$, and $\vx(T) \in G$.
\end{definition}
In other words, whenever the system is initialized inside the set $I$,
it stays inside the safe set $S$ until it reaches the goal set $G$.
Alternatively, we may express the specification in temporal logic as
$I\ \implies\ (S\ \scr{U}\ G)$, where $\scr{U}$ is the temporal operator
``until".

We will assume that set $S$ is a compact basic semialgebraic set. Typical examples include polytopes defined by linear inequalities or ellipsoids, that can be easily checked for the properties such as compactness and nondegeneracy. Also,  sets $I$ and $G$ are compact semialgebraic sets.

The special case when $I = S$ will be called \emph{uninitialized}
RWS. Such a property simply states that the system initialized inside
the set $S$ continues to remain in $S$ until it reaches a goal state
$\vx \in G$ at some finite time instant $T$. This case is suitable for
building a finite abstraction as mentioned in Sec.~\ref{sec:intro}.

\subsection{Control Certificates}
Encoding verification and synthesis problems into (control) certificates, which are defined by a set of conditions, is a standard approach. For example Lyapunov functions have been used for ensuring stability and barrier functions are employed to reason about safety properties.
However, these functions are not usually known in advance. To discover such a function in the first place, we solve a constrained problem in which certificates are parameterized. Usually, certificates are defined over polynomials with unknown coefficients and the problem reduces to finding proper coefficients for  polynomials~\cite{prajna2002introducing,dimitrova2014deductive}. For example, to find a Lyapunov function, first, a template for Lyapunov function $V$ is chosen: $V = \sum_i \vc_\alpha \vx^\alpha$, where $\vx^\alpha$ is a monomial with degree greater than zero. Then, solving the following constrained problem yields a Lyapunov function for proving stability to origin: $(\exists \vc) \ (\forall \vx \neq 0) \left( V(\vx) > 0 \land \dot{V}(\vx) < 0	\right)$, where $\dot{V}$ is $\nabla V.f(\vx)$. In these techniques, it is essential to define \emph{control} certificate with a simple structure that can be discovered automatically. In the subsequent we combine the certificates for safety and liveness to obtain a certificate for RWS properties.

\section{RWS for Basic Semialgebraic Safe Sets}\label{sec:rws-basic}
In this section, we first focus on the uninitialized RWS problem
($I=S$) and provide solutions for the case when $S$ is a basic
nondegenerate semialgebraic set (see
Def.~\ref{Def:basic-semialgebraic-set}). 

Let $S$ be a nondegenerate basic semialgebraic sets, as in
Def.~\ref{Def:basic-semialgebraic-set}.  Let $\partial S$ be partitioned into nonempty facets $F_{1}, \ldots,$ $ F_{l_k}$. Each
facet $F_{k}$ is, in turn, defined by two sets of polynomial
inequalities $F_{k}^{<}$ of inactive constraints and $F_{k}^{=}$ of
active constraints: $F_{k}=\{ \bigwedge_{p_{S,j} \in F_{k}^{<}} p_{S,j}(\vx) < 0\ \land\ \bigwedge_{p_{S,j} \in F_{k}^{=} } p_{S,j}(\vx) =  0 \}$.

For each state on a facet and not in $G$, we require the existence of a mode $q$,
whose vector field points inside $S$. Additionally, we will require a
certificate $V$ to decrease everywhere in $S \setminus G$. For any polynomial
$p$, let $\dot{p_q}:\ (\nabla p)\cdot f_q(\vx)$. By combining conditions for safety and liveness, one can obtain the following conditions:

\begin{equation}\label{eq:rws-basic}
\begin{cases}
\vx \in int(S)\setminus G \implies (\exists\ q)\ \dot{V_q}(\vx) < - \epsilon  \\[4pt]
\vx \in F_{1} \setminus G \implies (\exists\ q)
\left( \begin{array}{c}
	 \dot{V_q}(\vx) < - \epsilon\ \land\ \bigwedge  \limits_{p \in F_{1}^{=}} \begin{array}{c} \dot{p_q}(\vx) <- \epsilon\end{array} \end{array} \right)  \\
\vdots\\
\vx \in F_{l_k} \setminus G \implies (\exists\ q)
\left(\begin{array}{c} 
\dot{V_q}(\vx) < - \epsilon\ \land\ \bigwedge  \limits_{p \in F_{l_k}^{=}} \begin{array}{c} \dot{p_q}(\vx) <- \epsilon \end{array}  \end{array} \right)
\,.
\end{cases}
\end{equation}

The first condition in Eq.~\eqref{eq:rws-basic} states that $V$ must strictly decrease
everywhere in the set $int(S) \setminus G$. The subsequent conditions
treat each facet $F_j$ of the set $S$ and posit the
existence of a mode $q$ for each state that causes the active
constraints and the function $V$ to decrease. 

However, we note that as the number of state variables increases, the
number of facets can be exponential in the number of inequalities that
define $S$~\cite{helwa2013monotonic} . This poses a serious limitation to the applicability of Eq.~\eqref{eq:rws-basic}.

Our solution to this problem, is based partly on the idea of
exponential barriers discussed by Kong et
al.~\cite{kong2013exponential}. Rather than force the vector field to
point inwards at each facet, we simply ensure that each polynomial
inequality $p_{S,j} \leq 0$ that defines $S$, satisfies a decrease
condition outside set $G$. Thus, Eq.~\eqref{eq:rws-basic} is
replaced by a simpler (relaxed) condition:

\begin{equation}\label{eq:rws-basic-exp-temp}
\begin{array}{l}
\vx \in S \setminus G \implies (\exists\ q)  \ \ \dot{V_q}(\vx) < -\epsilon \ \ \land \bigwedge_j\ \left( \begin{array}{c}(\dot{p_{S,j,q}}(\vx) + \lambda\ p_{S,j}(\vx))\ < -\epsilon
	\end{array} \right)\,.
\end{array}
\end{equation}
Here $\lambda > 0$ is a user specified parameter. This rule is a
relaxation of ~\eqref{eq:rws-basic}.  The rule is made
stronger for larger values of $\lambda$. However, larger values
of $\lambda$ can cause numerical difficulties in practice while
searching for a control certificate.

For safety constraints, we require $\dot{p}_q$ to be numerically
$\leq -\epsilon$ mainly, to avoid numerical issues. This can be restrictive for cases where $\dot{p_{S,j,q}}$ is simply zero. To go around this, we define a set of facets $J_q = \{j | (\exists x) \ \dot{p_{S,j,q}}(x) > 0\}$ for each mode $q$. Informally speaking, $J_q$ is set of all facets for which change of $p_{S,j}$ must be considered when mode $q$ is selected. Because for each facet $j \notin J_q$, $p_{S,j}$ will never increase as long as mode $q$ is selected. Then, the conditions become:
\begin{equation}\label{eq:rws-basic-exp}
\begin{array}{l}
\vx \in S \setminus G \implies (\exists\ q)  \ \ \dot{V_q}(\vx) < -\epsilon \ \ \land \bigwedge_{j \in J_q} \ \left( (\dot{p_{S,j,q}}(\vx) + \lambda\ p_{S,j}(\vx))\ < -\epsilon \right)\,.
\end{array}
\end{equation}


As mentioned earlier, the problem of control synthesis consists of two
phases. The first phase deals with the problem of \emph{finding a
control certificate} $V(\vx)$ that
satisfies~\eqref{eq:rws-basic-exp}. We use
a \emph{counter-example guided inductive synthesis} (CEGIS) framework
to find such certificates. In the second phase, a switching strategy is
extracted from the control certificate to design the final controller.
We now examine each phase, in turn.

\subsection{Discovering Control Certificates}

We now explain the CEGIS
framework that searches for a suitable control certificate $V$. To
synthesize a control certificate, we start with a parametric form
$V_\vc(\vx) = V(\vc,\vx): \sum_{i=1}^N c_i g_i(\vx)$ with some 
(nonlinear) basis functions
$g_1(\vx),$ \ldots $, g_N(\vx)$ chosen by the user, and unknown
coefficients $\vc:\ (c_1,$ \ldots $, c_N)$, s.t. $\vc \in C$ for a compact set $C \subseteq \reals^N$. The certificate $V$ is a linear function over $\vc$.

The constraints from Eq.~\eqref{eq:rws-basic-exp} become as follows:
\begin{equation}\label{eq:template-rws-basic-exp}
\begin{array}{l}
(\exists \vc \in C) \ (\forall \vx \in X) \ \vx \in S \setminus G \implies \bigvee_q 
\left(
\dot{V_q}<-\epsilon\land
 \bigwedge_{j \in J_q} \left( \dot{p_{S,j,q}}(\vx) + \lambda p_{S,j}(\vx)<-\epsilon \right)
\right) \,.
	
\end{array}
\end{equation}

The constraints in Eq.~\eqref{eq:template-rws-basic-exp} has a complex
quantifier alternation structure involving the $\exists \vc$
quantifier nested outside the $\forall \vx$ quantifier.  First, we note that
$J_q$ is computed separately and here we assume it is given. Next, we modify an
algorithm commonly used for program synthesis problems to the problem
of synthesizing the coefficients $\vc \in C$~\cite{solar2008program}.

The counterexample guided inductive synthesis (CEGIS) approach has its
roots in program synthesis, wherein it was proposed as a general
approach to solve $\exists \forall$ constraints that arise in such
problems~\cite{solar2008program}.  The key idea behind the CEGIS
approach is to find solutions to such constraints while using a
satisfiability (feasibility) solver for quantifier-free formulas that
check whether a given set of constraints without quantifiers have a
feasible solution.

Solvers like Z3 allow us to solve many different classes of
constraints with extensive support for linear arithmetic
constraints~\cite{DeMoura+Bjorner/08/Z3}.  On the other hand, general
purpose nonlinear delta-satisfiability solvers like dReal, support the
solving of quantifier-free nonlinear constraints involving
polynomials, trigonometric, and rational
functions~\cite{DBLP:conf/cade/GaoKC13}. However, the presence of quantifiers
drastically increases the complexity of solving these
constraints. Here, we briefly explain the idea of CEGIS procedure for
$\exists\forall$ constraints of the form
\[ (\exists\ \vc \in C)\ (\forall\ \vx \in X)\ \Psi(\vc,\vx) . \]
Here, $\vc$ represents the unknown coefficients of a control
certificate and $\vx$ represents the state variables of the system.
Our goal is to find one witness for $\vc$ that makes the overall
quantified formula true. The overall approach constructs, maintains,
and updates two sets iteratively:
\begin{enumerate}

\item $X_i \subseteq X $ is a finite set of \emph{witnesses}. This is explicitly represented as $X_i = \{ \vx_1, \ldots, \vx_i \}$.
\item $C_i \subseteq C$ is a (possibly infinite) subset of  available \emph{candidates}. This is implicitly represented by a constraint $\psi_i(\vc)$, s.t.
$C_i:\ \{ \vc\ \in C\ |\ \psi_i(\vc) \}$.
\end{enumerate}

In the beginning, $X_0 = \{ \}$
and $\psi_0:\ \true$ representing the set $C_0:\ C$.

At each iteration, we perform the following steps:

\noindent (a) \emph{Choose} a
candidate solution $\vc_{i+1} \in C_i$. This is achieved by checking
the feasibility of the formula $\psi_i$. Throughout this paper, we
will maintain $\psi_i$ as a \emph{linear arithmetic} formula that
involves boolean combinations of linear inequality
constraints. Solving these problems is akin to solving linear
optimization problems involving disjunctive constraints. Although the
complexity is NP-hard, solvers like Z3 integrate fast LP solvers with
Boolean satisfiability solvers to present efficient solutions~\cite{DeMoura+Bjorner/08/Z3}.

\noindent (b) \emph{Test} the current candidate. This is achieved by
testing the satisfiability of $ \lnot \Psi(\vc, \vx)$ for fixed
$\vc = \vc_{i+1}$. In doing so, we obtain a set of nonlinear constraints
over $\vx$. We wish to now check if it is feasible.

If $\lnot \Psi(\vc_{i+1},\vx)$ has no feasible solutions, then
$\Psi(\vc_{i+1},\vx)$ is true (valid) for all $\vx$. Therefore, we can
stop with $\vc = \vc_{i+1}$ as the required solution for $\vc$.

Otherwise, if $\lnot \Psi(\vc, \vx)$ is feasible for some $\vx= \vx_{i+1}$,
we add it back as a witness: $X_{i+1}:\ X_i \cup \{ \vx_{i+1}\}$. The
formula $\psi_{i+1}$ is given by
\[ \psi_{i+1}:\ \psi_i\ \land\ \Psi(\vc, \vx_{i+1}) \,.\]
Note that $\psi_{i+1}\ \implies\ \psi_i$, and $\vc_{i+1}$ is no longer
a feasible point for $\psi_{i+1}$.  The set $C_{i+1}$ described by
$\psi_{i+1}$ is:
\[ C_{i+1}:\ \{ \vc \in C\ |\ \Psi(\vc,\vx_i)\ \mbox{holds for each}\ \vx_i \in X_{i+1} \} \,.\]

The CEGIS procedure either (i) runs forever, or (ii) terminates after
$i$ iterations with a solution $\vc: \vc_{i}$, or (iii) terminates
with a set of witness points $X_i$ proving that no solution exists. 

We now provide further details of the CEGIS procedure adapted to find
a certificate that satisfies Eq.~\eqref{eq:template-rws-basic-exp}. In
the CEGIS procedure, the formula $\Psi(\vc, \vx)$ will have the following form:
\begin{equation}
\label{eq:phi}
\begin{cases}
	\vx \in R_1 \implies \varphi_1(\vc, \vx) \\
	\vx \in R_2 \implies \varphi_2(\vc, \vx) \\
	... \\
	\vx \in R_{N_j} \implies \varphi_{N_j}(\vc, \vx) \,,
\end{cases}
\end{equation}
and each $\varphi_j$ for $j = 1,\ldots, N_j$ has the form 
\begin{equation} \label{eq:phi-k-finite}
\bigvee_k \bigwedge_l \ p_{j,k,l}(\vc, \vx) > 0 \,,
\end{equation}
where $p_{j,k,l}(\vc, \vx)$ is a function linear in $\vc$ and possibly nonlinear in $\vx$, depending on the dynamics and template used for the control certificate. 
%

The CEGIS procedure involves two calls to solvers: (a) Testing
satisfiability of $\psi_i(\vc)$ and (b) Testing the satisfiability of
$ \lnot \Psi(\vc_{i+1}, \vx)$. We shall discuss each of these problems
in the following paragraphs.

\paragraph{Finding Candidate Solutions:}
Given a finite set of witnesses $X_i$, a solution exists for
$\psi_{i+1}$ iff there exists $\vc \in C$ s.t. 
\begin{equation*} \bigwedge_{\vx \in
X_i} \ \bigwedge_{j=1}^{N_j} \left( \vx \in R_j \implies \bigvee_k \bigwedge_l
p_{j,k,l}(\vc, \vx) > 0 \right)\,, \end{equation*} and since $p_{j,k,l}$ is a linear function in $\vc$, such $\vc$ can be found by solving a formula in Linear Arithmetic Theory ($\scr{LA}$).

\paragraph{Finding Witnesses:} Finding a witness for a given candidate solution $\vc_i$ involves checking the satisfiability $\lnot \Psi$. Whereas $\Psi$ is a
conjunction of $N_j$ clauses, $\lnot \Psi$ is a disjunction of
clauses.  The $j^{th}$ clause in $\lnot \Psi$ ($1 \leq j \leq N_j$) has the form 
\begin{equation} \label{eq:ce-piece} \vx \in
R_j \land \bigwedge_k \bigvee_l p_{j,k,l}(\vc_i, \vx) \leq 0 \, . \end{equation}

We will test each clause separately for satisfiability.
Assuming that $p_{j,k,l}$ is a general nonlinear function over $\vx$, SMT solvers like dReal~\cite{DBLP:conf/cade/GaoKC13} can be used to solve this over a compact set $R_j$.  Numerical SMT
solvers like dReal can either conclude that the given formula is
unsatisfiable or provide a solution to a ``nearby'' formula that is
$\delta$ close. The parameter $\delta$ is adjusted by the user. As a result, dReal can correctly conclude that the current
candidate yields a valid certificate.  On the other hand, its witness
may not be a witness for the original problem. In this case, using the
spurious witness may cause the CEGIS procedure to potentially continue
(needlessly) even when a solution $\vc_i$ has been
found. Nevertheless, the overall procedure produces a correct result
whenever it terminates with an answer.

\begin{example}
\label{ex:basic}
	This example is adopted from~\cite{nilssonincremental}. There
	are two variables and three control modes with the dynamics given
	below: \begin{align*} \begin{array}{c} \left[ \begin{array}{c} \dot{x_1} \\ \dot{x_2} \end{array} \right]
	= \left[ \begin{array}{c} -x_2-1.5x_1-0.5x_1^3 \\
	x_1 \end{array} \right] + B_q ,\  B_{q1}
	= \left[ \begin{array}{c} 0 \\ -x_2^2 +
	2 \end{array} \right],\  B_{q2} = \left[ \begin{array}{c}
	0 \\ -x_2 \end{array} \right],\  B_{q3}
	= \left[ \begin{array}{c} 2 \\
	10 \end{array} \right] \,.  \end{array} \end{align*} \\ The
	goal is to reach the target set $G: (x_1+0.75)^2 +
	(x_2-1.75)^2 \leq 0.25^2$, a circle centered at $(-0.75,1.75)$,
	as shown in Figure~\ref{fig:basic}, while staying in the
	safe region given by the rectangle $S_0: [-2,2] \times
	[-2,3]$:
	\[
	S_0 : \{ \vx | (x_1 + 2)(x_1 - 2) \leq 0 \land (x_2 + 2)(x_2 - 3) \leq 0 \} \,.
	\]

	First, we shift co-ordniates to transform $(-0.75, 1.75)$ as
	the new origin. Then, we use a quadratic template for V
	($c_1 x_1^2 + c_2 x_1 x_2 + c_3 x_2^2$) , $\epsilon = 1$,
	$\lambda = 5$. The solution $V$ is found in $5$
	iterations. Then, we translate the
	function back to the original co-ordinates:
        \begin{align*}  
        V(x_1,x_2): & 37.782349 x_1^2 -
	2.009762 x_1 x_2 + 60.190607 x_1 + 4.415093 x_2^2 -
	16.960145 x_2 + 37.411604 \,.
        \end{align*}
\end{example}

\begin{example}
	A unicycle~\cite{zamani2012symbolic} has three variables. $x$ and $y$ are position of the car and $\theta$ is its angle. The dynamics of the system is
	$\dot{x} = u_1 cos(\theta) \,, \ 
		\dot{y} = u_1 sin(\theta) \,, \ 
		\dot{\theta} = u_2$,
	where $u_1$ and $u_2$ are inputs. Assuming a switched system, we consider $u_1 \in \{-1, 0, 1\}$ and $u_2 \in \{-1, 0, 1\}$. The safe set is $[-1, 1]\times[-1, 1]\times[-\pi, \pi]$ and the target facet is $x = 1$.
	We use a template that is linear in $(x,y)$ and quadratic in $\theta$. Using $\epsilon = 0.1$ and $\lambda = 0.5$, the following CLF is found after $22$ iterations:
	\[ V(\vx) : - x - y - 0.5881 \theta +\theta^2 - 0.1956 \theta x + \theta y \,.\]
\end{example}

\begin{example}
\label{ex:basic-4D}
	This example is adopted from~\cite{habets2004control}. There are four variables and two control inputs. The dynamic is as follows:
	\begin{align*}
	\begin{array}{c}
		\left[ \begin{array}{c}
 			\dot{x_1} \\ \dot{x_2} \\ \dot{x_3} \\ \dot{x_4}
 		\end{array} \right]
 		 = \left[ \begin{array}{c}
			x_1 + x_2 + 8\\
			-x_2 + x_3 + 1 \\
			-2x_3 + 2x_4 + 1 \\
			-3x_4 + 1
		\end{array} \right] + 
		\left[  \begin{array}{c}
			u_1 \\ -u_2 \\ -2u_1 \\ u_2
		\end{array} \right] \,.
	\end{array}
	\end{align*} \\	
	The region of interest $S$ is hyber-box $[-1, 1]^4$ and the input belongs to set $[0, 1] \times [0, 2]$. The goal is to reach facet $x_1 = 1$, while staying in $S$ as the safe region. 
	
	First, we discretize the control input to model the system as a switched system. For this purpose, we assume $u_1 \in \{0, 1\}$ and $u_2 \in \{0, 0.5, 1, 1.5, 2\}$. Then, we use a linear template for the CLF ($c_1 x_1 + c_2 x_2 + c_3 x_3 + c_4 x_4$) , $\epsilon = 0.1$, $\lambda = 5$. CEGIS framework finds certificate $V(\vx): -0.13333344 (x_1+x_2+x_3+x_4)$.
\end{example}

\subsection{Control Design}
Thus far, we discussed the CEGIS framework for finding a control
certificate. Extracting the $\ctrl$ function from the certificate is
now considered.  Given a control certificate $V$ satisfying
Eq.~\eqref{eq:rws-basic-exp}, the choice of a switching mode is
dictated by a function $\eta_q(\vx)$ defined for each state $\vx \in X$
and mode $q \in Q$ as follows:
\[
\eta_q(\vx):\ \max\left(
\begin{array}{c}
\dot{V_q}(\vx),\ 
\eta_{S,1,q}(\vx),\ \cdots, \eta_{S,k,q}(\vx)
\end{array}
\right)\,,
\]
where for all $j \in J_q$, $\eta_{S,j,q}$ is $\dot{p_{S,j,q}} +\lambda p_{S,j}$ and for $j \notin J_q$, $\eta_{S,j,q} = -\infty$ or equivalently $\eta_{S,j,q} = -L$ for some large constant $L$.

The idea is that whenever (at time $t$) the controller chooses a mode
$q$ s.t. $\eta_q(\vx(t)) < -\epsilon$, one can guarantee that
$\eta_{q}(\vx(t))$ $ < 0$ holds for all $t \in [T, T+\delta)$, for some
minimum time $\delta$. Therefore, for some fixed $\epsilon_s$ ($0 < \epsilon_s < \epsilon$), the function $\ctrl$ for any $\vx \in S \setminus G$ can be
defined as
\begin{equation}
	\label{eq:controller}
	\ctrl(q, \vx) := \begin{cases} 
	\hat{q} & \mbox{if}\ 
		\left(\begin{array}{c} 
		\eta_{q}(\vx) \geq - \epsilon_s
		\land \ \eta_{\hat{q}}(\vx) < - \epsilon 
		\end{array} \right) \\ \\
	q &  \mbox{otherwise} \,.
	\end{cases}
\end{equation}

In other words, the controller state persists in a given mode $q$
until $\eta_q(\vx) \geq -\epsilon_s$. Then,
given that $\vx \in S \setminus G$, Eq.~\eqref{eq:rws-basic-exp} will
provide us a new control mode $\hat{q}$ that satisfies $\eta_{\hat{q}}(\vx) <
-\epsilon$.  This mode is chosen as the next mode to switch to.

\begin{example}
Consider once again, the problem from Ex.~\ref{ex:basic}. Using the
defined function $V(x_1,x_2)$, Eq.~\eqref{eq:controller}
yields a controller.  Figure~\ref{Fig:closed-loop-basic-example} shows
some of the simulation traces of this closed loop system, demonstrating the
RWS property.

\end{example}

\begin{figure}[H]
\begin{center}
\begin{subfigure}{.4\textwidth}
\centering
\includegraphics[width=1\textwidth]{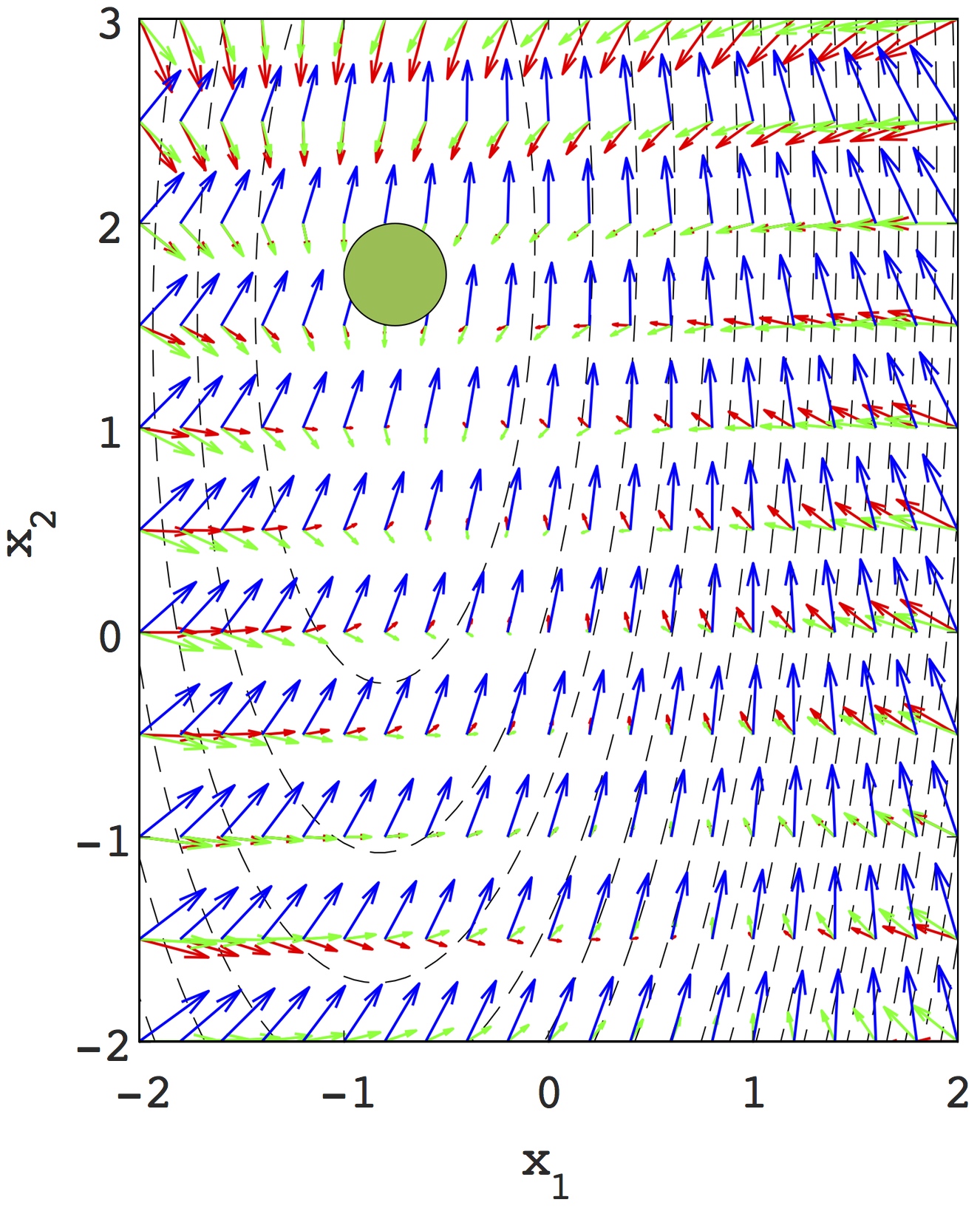}
\caption{}
\label{fig:basic}
\end{subfigure}
\qquad
\begin{subfigure}{.4\textwidth}
\centering
\includegraphics[width=1\textwidth]{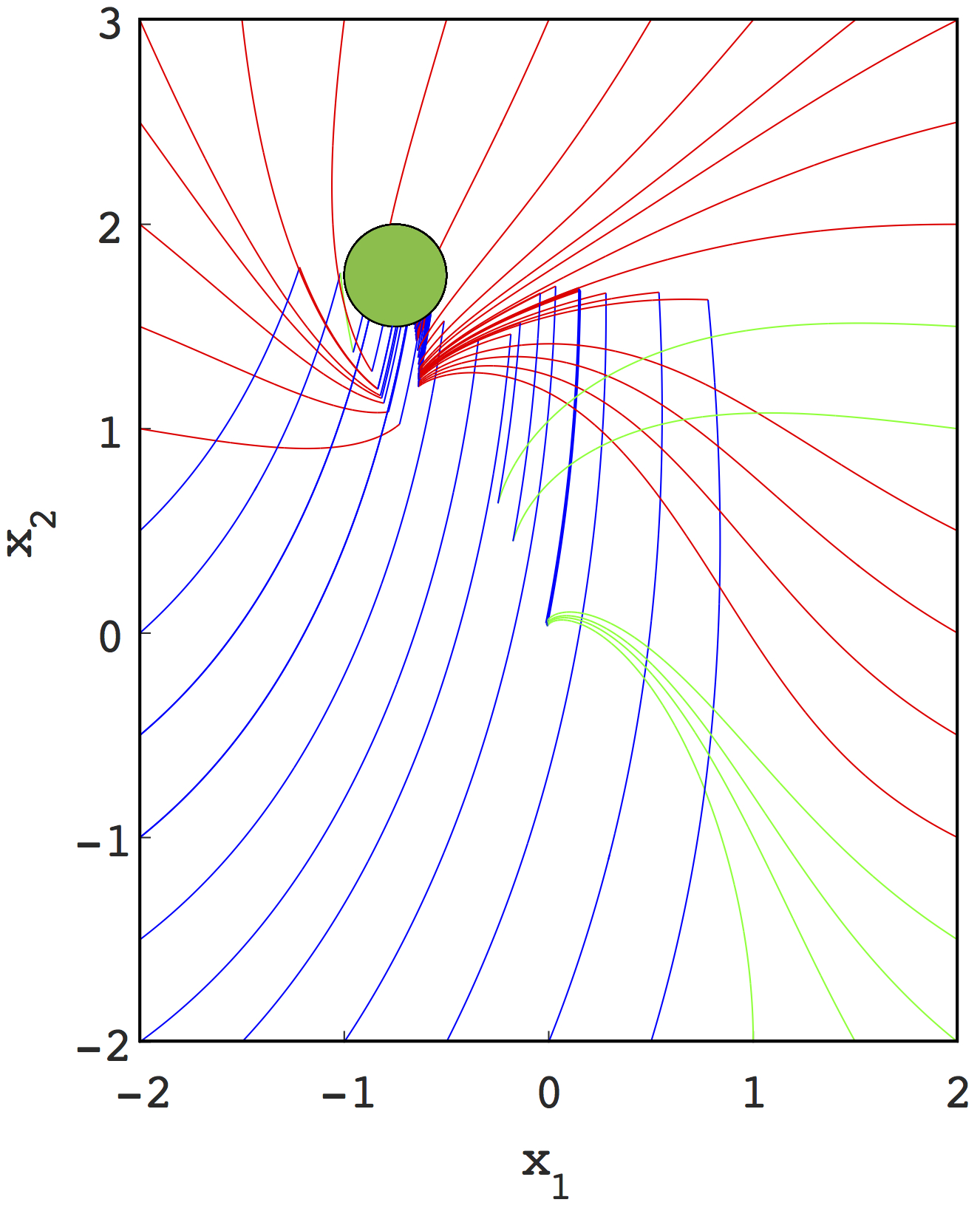}
\caption{}
\label{Fig:closed-loop-basic-example}
\end{subfigure}
\caption{(a) Region $G$ for Example~\ref{ex:basic} is shown shaded in the center, and the vector fields for modes $q_1,q_2$ and $q_3$ are shown in red, green and blue, respectively. Level-sets of $V$ are shown with black dashed lines. (b) Closed loop trajectories for Example~\ref{ex:basic} using the controller defined by Eq.~\eqref{eq:controller}. The segments shown in colors red, green and blue correspond to the modes $q_1, q_2$ and $q_3$, respectively.}
\end{center}
\end{figure}

We now establish the key result that provides a minimum dwell time guarantee.
\begin{lemma}
\label{lem:delta-exists}
There exists a $\delta > 0$ s.t. for all initial conditions $\vx(T) \in
	S\setminus G$, if $\eta_{q}(\vx(T)) < - \epsilon$, and if the mode
	of the system is set to $q$ at time $T$, then
\[(\forall t \in [T, T+\delta]) \ (\vx(t) \in
	S \setminus G)\ \implies\ \eta_{q}(\vx(t)) \leq
	- \epsilon_s \,.\]
\end{lemma}

\begin{proof}
	Let $T + \delta$ be the earliest time instant,  where $\eta_{q}(\vx(T+\delta)) \geq - \epsilon_s$ while at the same
        time
	\[ 
		(\forall t \in [T , T + \delta ])\  q(t) = q,\ \vx(t) \in S \setminus G \,.
	\]
	At time $T$, $\eta_q(\vx(T)) < -\epsilon$ and at time $T+\delta$, $\eta_q(\vx(T+\delta)) = -\epsilon_s$.
        Note that $\eta_q(\vx)$ is defined as $\max(\alpha_1(\vx),\ldots,\alpha_{m}(\vx))$ for some smooth
        functions $\alpha_1,\ldots, \alpha_m$. As a result, 
	Since $S$ is a bounded set, and $p$, $f_q$, and $V$ are bounded over $S$,  there exists a constant $\Lambda > 0$ s.t.
	\begin{equation}
	\label{eq:dg-leq-epsilon}
	(\forall \ \vx \in S) \ \dot{\alpha_{i,q}}\leq \Lambda \,.
	\end{equation}
	Therefore, for each $\alpha_i$, we have
        \begin{align*}
        \alpha_i(\vx(T+\delta)) & = \alpha_i(\vx(T)) + \int_{t=T}^{T+\delta} \dot{\alpha_{i,q}}(\vx(t)) dt
        \leq \alpha_i(\vx(T)) + \Lambda \delta \,.
        \end{align*}
        As a result, we conclude that
        \begin{align*}
        \eta_q(\vx(T+\delta)) &= \max_{i} \alpha_i(\vx(T+\delta)) = \alpha_{j^*}(\vx(T+\delta))
        \leq \alpha_{j^*}(\vx(T)) + \Lambda \delta \leq \eta_q(T) +\Lambda \delta\,.
        \end{align*}
	Therefore, we can conclude $- \epsilon_s < -\epsilon + \Lambda \delta \implies  \frac{\epsilon - \epsilon_s}{\Lambda} < \delta$ and there exists a fixed $\delta >  \frac{\epsilon - \epsilon_s}{\Lambda}  > 0 $ s.t.
	\[
	(\forall t \in [T, T+\delta)) \ \eta_{q}(\vx(t)) < -\epsilon_s \,. \blacksquare
	\]
\end{proof}

Eq.~\eqref{eq:controller} gives a switching strategy which respects the min-dwell time and as long as $\vx(t) \in S$, the controller guarantees $\eta_{q(t)}(\vx(t)) \leq - \epsilon_s$. I.e. for all $j \in J_q$, $\dot{V_q}(\vx(t)) \leq -\epsilon_s$ and $\dot{p_{S,j,q}}(\vx(t)) + \lambda p_{S,j}(\vx(t)) \leq -\epsilon_s$.

\begin{theorem}
\label{thm:simple}
	Given nondegenerate \ basic \ semialgebraic set $S$, a semialgebraic set $G$, and a function $V$ (satisfying Equation~\eqref{eq:rws-basic-exp}), the  control strategy defined by Eq.~\eqref{eq:controller}
        respects the  min-dwell time property and guarantees the RWS property defined by $S,G$:\ $S \implies S \scr{U} G$.
\end{theorem}

\begin{proof}
	As discussed, there exists a controller which respects
	the min-dwell time property. Also, the controller guarantees
	$\dot{V_q}(\vx) \leq -\epsilon_s$ and
	$(\dot{p_{S,j,q}}(\vx) + \lambda p_{S,j}(\vx) \leq
	-\epsilon_s$ (for all $j \in J_q$), as long as $\vx \in
	S \setminus G$.
	
	Assume $\vx(t)$ is on the boundary of $S$ (and not in $G$) at some time $t$. Because $S$ is assumed to be a nondegenerate basic semialgebraic set, there exists at least one $j$ s.t. $p_{S,j}(\vx(t)) = 0$. If $j \notin J_q$, by definition, $\dot{p_{S,j,q}}$ is negative for all states and  $p_{S,j,q}$ remains $\leq 0$ as long as mode $q$ is selected. 
	Otherwise ($j \in J_q$), we
	obtain $\dot{p_{S,j,q}}(\vx(t)) \leq -\epsilon_s < 0$.
	Therefore, there exits $\tau_j > 0$, s.t. $s \in (t, t
	+ \tau_j)$, we conclude that $p_{S,j,q}(\vx(s)) < 0$. As a
	result, the trajectory cannot leave the set $S$.

     Thus, the trace cannot leave $S$, unless it reaches $G$.  Now,
	we show that the trajectory cannot stay inside $S \setminus G$
	forever. By the construction of the controller, we can
	conclude time diverges (because the controller respects the
	min-dwell time property) and that $V$ decreases ($\dot{V_q}(\vx(t)) \leq -\epsilon_s$). However, the
	value of $V$ is bounded on bounded set $S \setminus
	G$. Therefore, $\vx$ cannot remain in $S \setminus G$ and the
	only possible outcome for the trace is to reach $G$. $\blacksquare$
\end{proof}

\section{RWS for Semialgebraic Safe Set}\label{sec:rws-semialgebraic}
As Habets et al~\cite{habets2006reachability} discussed, control-to-facet problems can be used to build an abstraction. Here, we demonstrate that the method described so far can be integrated in this framework to tackle more complicated problems. First, we briefly explain how the method works. For a more detailed discussion, the reader can refer to~\cite{habets2006reachability} or ~\cite{kloetzer2008fully}.

First, state space is decomposed into polytopes according to the specifications. Here, we can use basic semialgebraic sets instead of polytopes. Then, for each such a set $u$, we consider an abstract state $\ab{u}$. Furthermore, for each of its $n-1$ dimensional facet $F$, a control-to-facet problem is solved. The corresponding problem is to find a control strategy to reach $F$ starting from $u$. If the control-to-facet problem is solved successfully, then for each basic semialgebraic set $v$ with a $n-1$ dimensional facet $F' \subseteq F$, an edge from $\ab{u}$ to $\ab{v}$ (with label/action $F$) is added to the abstraction. Also for each basic semialgebraic set $u$, one can check if $u$ is a control invariant to build self loops. However, for RWS properties, self loops are redundant and we skip them here. After building the abstract system, we use standard techniques to solve the problem for finite systems. If the problem could be solved for the abstract system, then, one can design a controller.

First, for each abstract state $\ab{u}$, there is at least one action $F$ that agrees with the winning strategy for the abstract system. Let that action be $\ac{\ab{u}}$. The idea is to implement transition $\ac{\ab{u}}$, using controller $\ctrl_{u,\ac{\ab{u}}}$ for the corresponding control-to-facet problem~\cite{habets2006reachability}. Formally, the controller can be defined as follows:

\begin{equation} \label{eq:controller-general}
	\ctrl(q, \vx) = \begin{cases}
		\ctrl_{u_1,\ac{\ab{u_1}}}(q, \vx)  \hspace{1cm} \vx \in u_1 \\
		\vdots \\
		\ctrl_{u_s,\ac{\ab{u_s}}}(q, \vx)  \hspace{1cm} \vx \in u_s \,.
	\end{cases}
\end{equation}

When $\vx$ belongs to multiple sets, one can break the tie by some ordering,
where states in the winning set have priorities. It is worth mentioning that
combining these controllers together, does not produce
any Zeno behavior as it is guaranteed that each abstract state is visited only
once for RWS properties. However, superdense switching is
possible as two facets of a polytope can get arbitrarily close.

If one is interested in LTL properties (not just reach-while-stay) or 
min-dwell time property, one
possible solution is to use \emph{fat} facets, where the target sets are $n$
dimensional goal sets. This extends the domain 
of the control-to-facet problem to adjacent basic semialgebraic sets as well.
Also, it allows the controller to continue using current sub-controller 
for some minimum time (if min-dwell time requirement is not met), before
changing the sub-controller (at the switch time).

\begin{example} \label{ex:obstacle-1} Consider again the system from Example~\ref{ex:basic}, with the addition of  some obstacles~\cite{nilssonincremental}. More precisely, as shown in Fig~\ref{fig:obstacle1-basic-a}, safe set is $S = S_0 \setminus (O_1 \cup O_2)$. First, the safe set is decomposed into four basic semialgebraic sets, which are shown with $R_0$ to $R_3$ in Fig.~\ref{fig:obstacle1-basic-a}.
{\makeatletter
\let\par\@@par
\par\parshape0
\everypar{}
\begin{wrapfigure}{r}{0.5\textwidth}
  \centering
  \begin{subfigure}{0.2\textwidth}
  \includegraphics[width=1\linewidth]{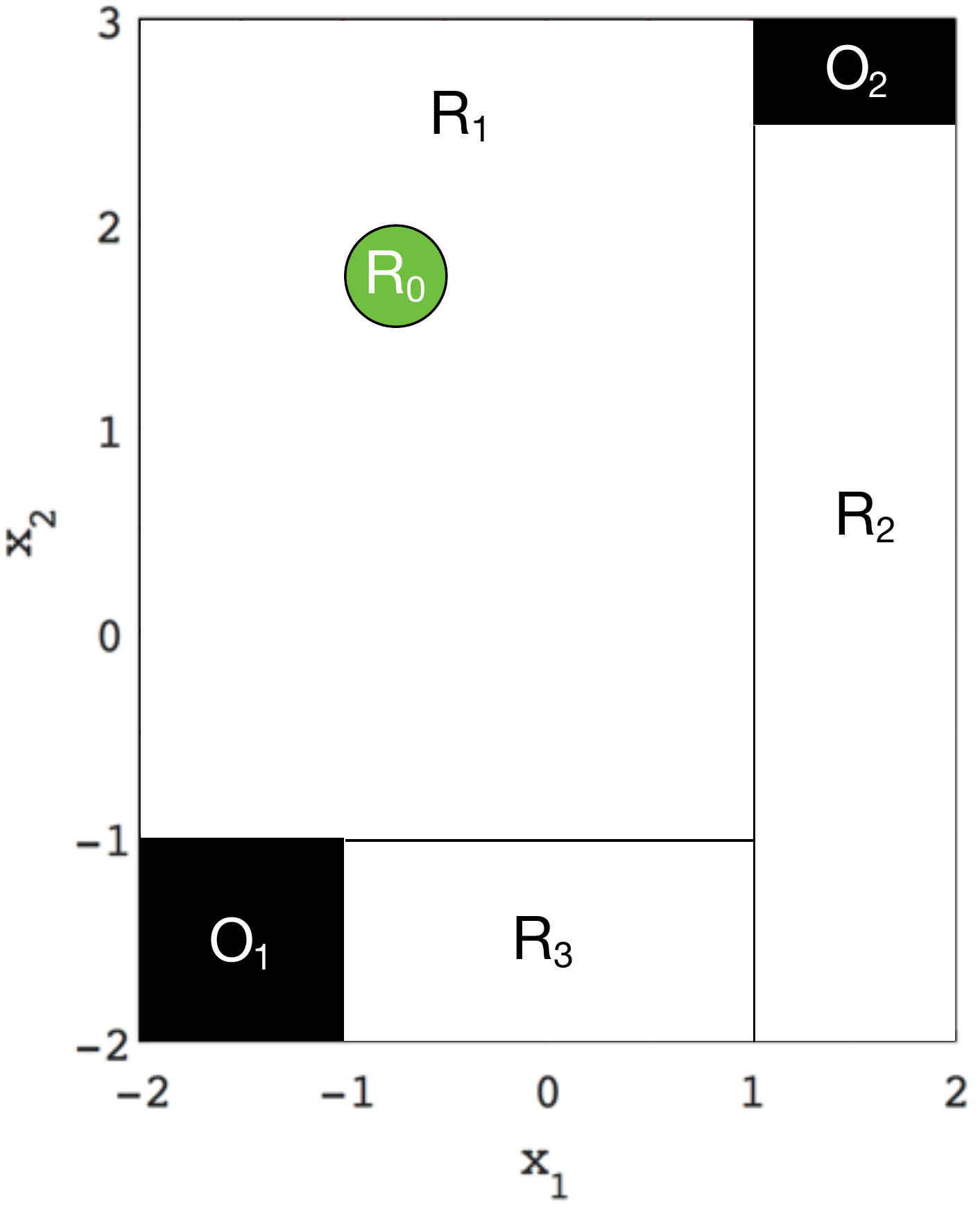}
  \caption{}
  \label{fig:obstacle1-basic-a}
  \end{subfigure} \qquad
  \begin{subfigure}{0.2\textwidth}
  \includegraphics[width=1\linewidth]{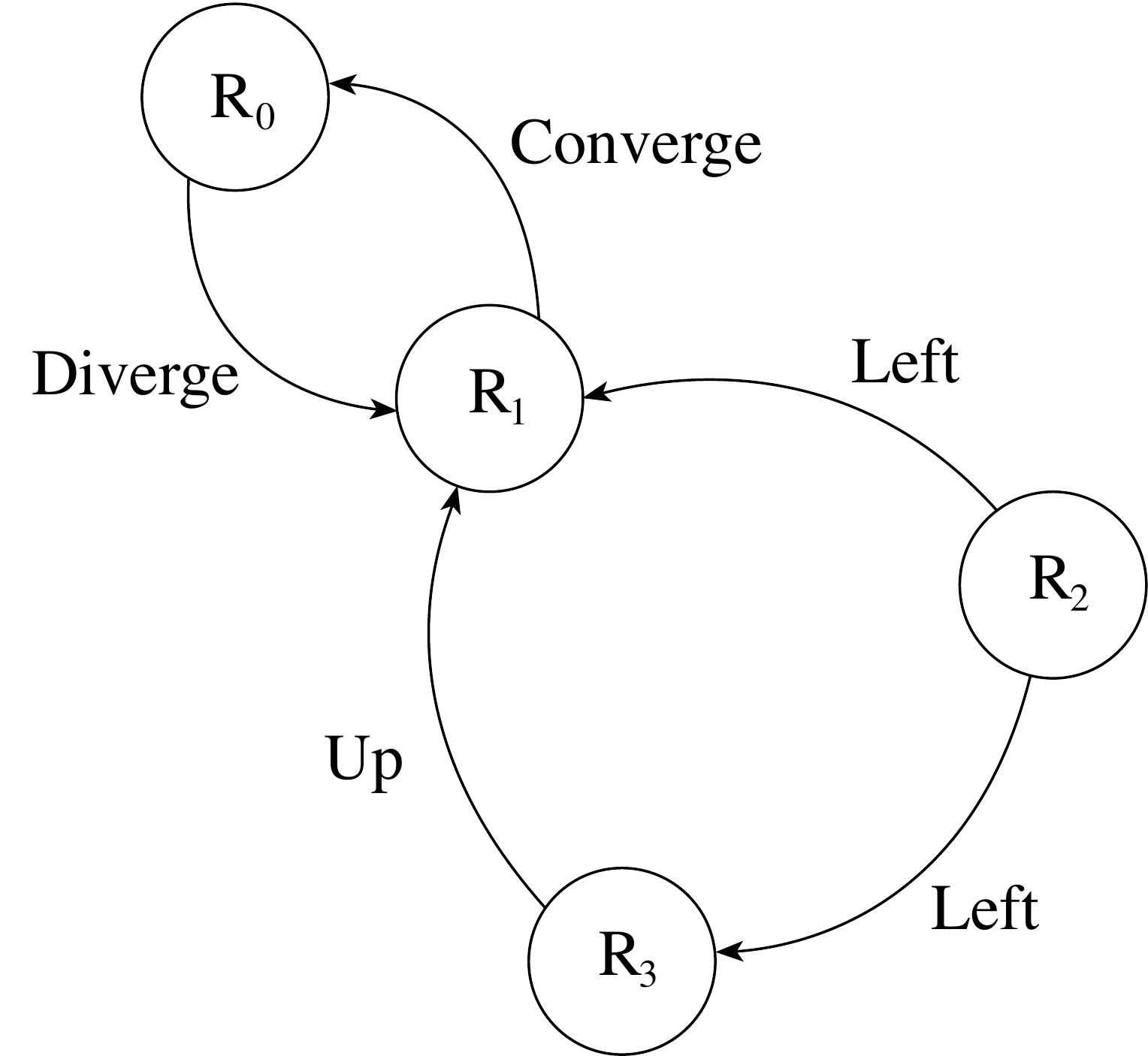}
  \caption{}
  \label{fig:obstacle1-basic-b}
  \end{subfigure}
\caption{(a) Schematic view of state decomposition. (b) Finite abstraction for the original problem.} 
\end{wrapfigure} 
$R_0$ is the target set. Next, we build a transition relation between four abstract states, representing four basic semialgebraic sets. This is done by solving seven RWS problems for basic semialgebraic sets. For $R_1$ to $R_0$, we use a quadratic template for $V$, and for other problems, we use linear template. The abstract system is shown in Fig.~\ref{fig:obstacle1-basic-b}. Next, the problem is solved for the abstract system. The solution to the abstract system is simple: if the state is in $R_2$, the controller uses the left facet to reach $R_1$ or $R_3$. Otherwise, if the state is in $R_3$, the controller uses the upper facet to reach $R_1$ and finally, if the state is in $R_1$, the controller makes sure the state reaches $R_0$.
\par}
\end{example}


\begin{example} \label{ex:unicycle}
This example is a path planning problem for the unicycle~\cite{rungger2016scots}. Projection of safe set on $x$ and $y$ yields a maze. The target set is placed at the right bottom corner of the maze (Fig.~\ref{fig:unicycle}).
 Using specification-guided technique, we modeled the system with 53 polyhedra. Each polyhedron is treated as a single state and a transition relation is built by solving 113 control-to-facet problems. Then, the problem is solved over the finite graph. The total computation took 1484 seconds. The figure also shows a single trajectory of the closed loop system.
\end{example}

\begin{figure}
\centering
\includegraphics[width=0.4\textwidth]
	{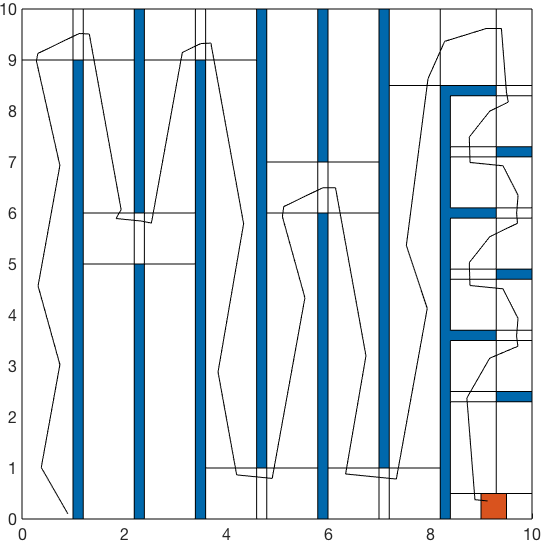}
\caption{Region $G$ is shown shaded in Orange, and unsafe regions are shown in blue. An execution trace of the car is shown for $x$ and $y$ variables.}
\label{fig:unicycle}
\end{figure}

\begin{example} \label{ex:unicycle-4D}
	This example is similar to Example~\ref{ex:unicycle}, except for the fact that there is no direct control over the angular velocity. More precisely, only the angular acceleration is controllable and the system would have the following dynamics
	$\dot{x} = u_{1} \cos (\theta),\,
	\dot{y} = u_{1} \sin (\theta),\,
	\dot{\theta} = \omega,\,
	\dot{\omega} = u_{2}\,.
	$
	Also, we assume $\omega \in [-1, 1]$. By changing the coordinates one can use $r = \sqrt{x^2 + y^2}$, $z_1 = x\cos (\theta)+y\sin (\theta)$ and $z_2 = y\cos (\theta)-x\sin (\theta)$ to define position and angle of the car(cf.~\cite{liberzon2012switching} for details). Then, we use the following template $V(x, y, \theta, \omega) = c_1 r^2 + c_2 z_1 + c_3 z_2\omega + c_4 \omega^2$,
	where the origin is located just outside of the target facet. Using this template, we find control certificates for all 113 control-to-facet problems in 5296 seconds.
\end{example}

\section{Initialized Reach-While-Stay}
So far, we discussed uninitialized RWS specifications ($S$ $ \implies
S \scr{U} G$). In these systems, we use boundary of safe set as
barrier.  However, as pointed out by Lin et al.~\cite{lin2007reachability},
this may not be the case. Now, we consider the initialized problem
for a given initial set $I$ ($I \implies S \scr{U} G$). To avoid 
technical difficulties,  we assume that $I \subseteq int(S)$. The
solution is to create a composite barrier that is formed by the
boundary of $S$ as well as other a priori unknown barrier functions.

\paragraph{Barrier Functions:} We recall that for a control barrier function~\cite{taly2010switching,xu2015robustness}, the following
conditions are considered 
\begin{equation} \label{eq:cbc-original}
\begin{array}{l}
\vx \in \partial S \implies B(\vx) > 0 \\
\vx \in I \implies B(\vx) < 0	\\
\vx \in S \implies \left( B(\vx) = 0  \implies (\exists q) \dot{B_q}(\vx) < -\epsilon \right) \,.
\end{array}
\end{equation}

This ensures that $B(\vx) = 0$ is a barrier and $\partial S$ is unreachable.
Eq.~\eqref{eq:cbc-original}, combined with the smoothness of $B$ and $f_q$ ensures
that as soon as the state is sufficiently ``close'' to the barrier,
 it is possible to choose a control mode that ensures the local
decrease of the $B$. 

The condition in Equation~\eqref{eq:cbc-original} can be
encoded into the CEGIS framework. However, the presence of the
equality $ B(\vx) = 0$ poses practical problems. In particular,
it requires for each candidate $B_\vc$, to find
a counterexample $\vx$ s.t. $B_\vc(\vx) \not= 0$. Unfortunately, such
an assertion is easy to satisfy, resulting in the procedure always
exceeding the maximum number of iterations permitted.

Again, we find that the following relaxation of the third condition is particularly effective in our experiments
\begin{equation}
\label{eq:cbc-safety}
	\begin{array}{l}
	\vx \in \partial S\ \implies\ B(\vx) > 0 \\
	\vx \in  I\ \implies\ B(\vx) < 0 \\
	\vx \in S \implies \bigvee_q \left( \begin{array}{c}
		\dot{B}_q(\vx) - \lambda B(\vx) < -\epsilon 
		\lor
		\dot{B}_q(\vx) + \lambda B(\vx) < -\epsilon
	\end{array} \right) \,,
	\end{array}
\end{equation}
for some constant $\lambda$.

Intuitively, by choosing $\lambda = 0$, the condition is similar to
that of Lyapunov functions, whereas as
$|\lambda| \rightarrow \infty$, the condition gets less conservative
and in the limit, it is equivalent to the original condition. In fact,
for smaller $|\lambda|$ CEGIS terminates
faster, but at the cost of missing potential solutions. On the other
hand, using larger $|\lambda|$, is less conservative at the cost of
CEGIS timing out. We also note that this formulation is less conservative than the one introduced by Kong et al.~\cite{kong2013exponential} as our formulation uses two exponential conditions which only forces decrease of value of $B$ around $B^* = \{\vx\  |\  B(\vx) = 0\}$.

To solve the RWS in general form, we define a finite set of barriers $\scr{B}$ with the following conditions:
\begin{equation} \label{eq:rws-general-1}
\begin{array}{l}
\vx \in \partial S \implies \bigvee_{B \in \scr{B}} B(\vx) > 0 \\
\vx \in I \implies \bigwedge_{B \in \scr{B}} B(\vx) < 0 \,.
\end{array}
\end{equation}

Also for each mode $q$, $\scr{B}_q$ is defined as $\scr{B}_q = \{B \in \scr{B}| (\exists \vx) \dot{B_q}(\vx) > 0\}$. Then, existence of a proper mode can be encoded as the following:

\begin{equation}\label{eq:rws-general-2}
\begin{array}{l}
	\vx \in S \setminus G \implies
	\left(
	\begin{array}{c}
	\left( \dot{V_q}(\vx) < -\epsilon \right) \ \land
	\bigwedge_{B \in \scr{B}_q}  \left( 
	\begin{array}{c}
	\dot{B_{q}}(\vx) + \lambda B(\vx) < -\epsilon \lor \\
	\dot{B_{q}}(\vx) - \lambda B(\vx) < -\epsilon	\end{array}
	\right)
	\end{array} 
	\right) \,.
\end{array}
\end{equation}

\begin{theorem}
\label{thm:simple}
	Given nondegenerate \ basic \ semialgebraic set $S$, semialgebraic sets $I$ and $G$, function $V$, and a non-empty set of functions $\scr{B}$ (satisfying Equation~\eqref{eq:rws-general-1} and ~\eqref{eq:rws-general-2}), there is a control strategy that respects the  min-dwell time property and guarantees the RWS property:\ $I \implies S \scr{U} G$.
\end{theorem}

To simplify these constraints and reduce the number of unknowns, one can use some of $p_{S,i}$'s to fix some of these barriers, which yields conditions similar to the ones used for the uninitialized problem. This trick is demonstrated in the following example.

\begin{example}\label{ex:dcdc}
	This example is taken from~\cite{mouelhi:hal-00743982}, in which a DC-DC converter is modeled with two variables $i$ and $v$. 
{\makeatletter
\let\par\@@par
\par\parshape0
\everypar{}
	\begin{wrapfigure}{r}{0.4\textwidth}
\begin{center}
\includegraphics[width=0.4\textwidth]
	{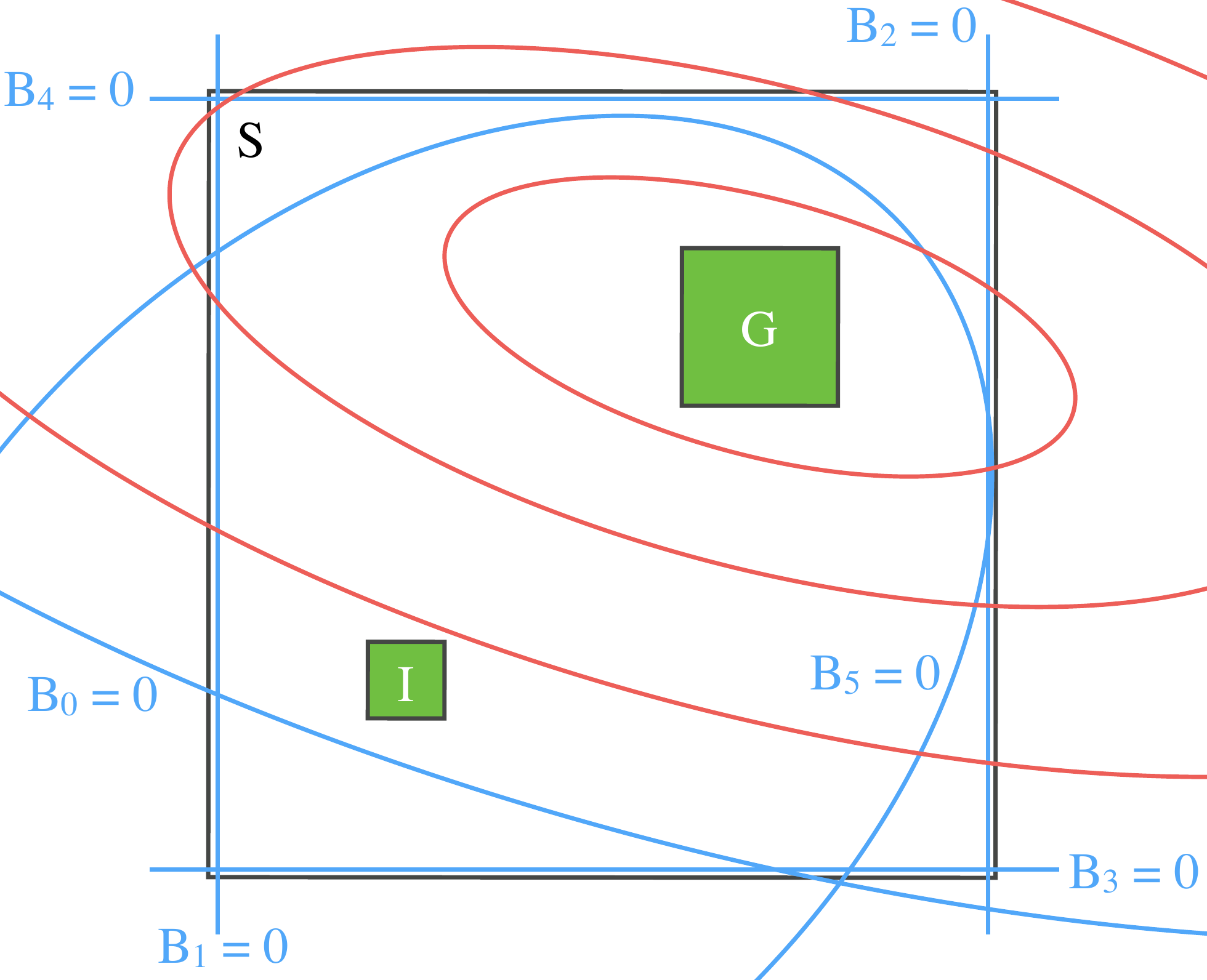}
\caption{The blue lines are the barriers and the red lines are level-sets of the Lyapunov function.}
\label{fig:ex-general}
\end{center}
\end{wrapfigure}
	The system has two modes $q_1$ and $q_2$, with the following dynamics:
	\begin{align*}
	q_1:& \begin{cases}
		\dot{i} = 0.0167 i + 0.3333 \\
		\dot{v} = -0.0142 v
	\end{cases} \\ 
	q_2:& \begin{cases}
		\dot{i} = - 0.0183 i -0.0663 v + 0.3333 \\
		\dot{v} = -0.0711 i -0.0142 v \,.
	\end{cases}
	\end{align*}
The safe set is $S: [0.65, 1.65] \times [4.95, 5.95]$ and the goal set is $G: [1.25, 1.45] \times [5.55, 5.75]$.
We assume initial set to be $I: [0.85, 0.95]\times[5.15, 5.25]$ (Fig.~\ref{fig:ex-general}).
Then, we use 5 barriers $B_0,\ldots, B_4$. Using boundaries of $S$, we choose $B_1,\ldots, B_4$ as follows:
\begin{align*}
B_1 =& 0.65 - i + \epsilon_b & \ 
&B_2 = 1.65 - i + \epsilon_b \\ 
B_3 =& 4.95 - v + \epsilon_b & \ 
&B_4 = 5.95 - v + \epsilon_b \,,
\end{align*}
\par}
where $\epsilon_b > 0$ is small enough that $I \subset int(\bigcap_{i=1}^{4} B_i)$. In this case, we choose $\epsilon_b = 0.01$. Notice that such $\epsilon_b$ always exists by the definition. Next, we assume $B_0 = V$ and both have the following template:
\[ B_0 = V : c_1 (i-1.35)^2 + c_2 (i-1.35) (v-5.65) + c_3 (v-5.65)^2 - 1 \,.\]
This template is chosen in a way that $V$ is a quadratic function with minimum value of $-1$ for the point of interest $i = 1.35$, $v = 5.65$. So far, we used these tricks to reduce the number of unknowns for barriers. However, our method fails to find a certificate. Next, we add one more barrier ($B_5$) to the formulation and we use the following template for $B_5 : c_4 (i-0.9)^2 + c_5(i-0.9)(v-5.2) + c_6(v-5.05)^2 - 1$, which is a quadratic function with minimum value of $-1$ for initial point $i = 0.9$, $v = 5.2$. This time, we can successfully find a control certificate. The final barriers and level-sets of the Lyapunov function is shown in Fig.~\ref{fig:ex-general}.
\end{example}

\paragraph{Comparison:} While abstraction based methods can provide a near optimal solution (are relatively complete), these methods can be computationally expensive. On the other hand, our method is a Lyapunov-based method and the solution is not necessarily (relatively) complete. For example, our approach assumes that control certificates with a given form (that is given as input by user) exist .  As such, the existence of such certificates is not guaranteed and thus, our approach lacks the general applicability of a fixed-point based synthesis. Also, for initialized problems our method needs an initial set as input, while for the abstraction based methods, maximum controllable region can be obtained without the need for specifying the initial set.
However, our method is relatively more scalable thanks to recent development in SMT solvers. Here, for the sake of completeness, we provide a brief comparison with SCOTS toolbox~\cite{rungger2016scots} for the examples provided in this article.
To compare Example~\ref{ex:basic-4D}, we use \emph{fat} facet and assume target set has a volume (otherwise, because of time discretization, SCOTS cannot find a solution). More precisely, we use target set $[1, 1.2]\times[-1, 1]^3$ instead of $[1,1]\times[-1,1]^3$. 

All the experiments are ran on a laptop with Core i7 2.9 GHz CPU and 16GB of RAM. The results are reported in Table~\ref{tab:scots}. We also note that if we use larger values for SCOTS parameters, SCOTS fails to solve these problems (initial set is not a subset of controllable region).  Table~\ref{tab:scots} shows that SCOTS performs much better for Example~\ref{ex:obstacle-1} and ~\ref{ex:dcdc} for which there are only $2$ state variables. For Example~\ref{ex:unicycle}, both methods have similar performances. And for Example.~\ref{ex:basic-4D} and Example.\ref{ex:unicycle-4D}, which have $4$ state variables, our method is faster.

\begin{table}[t!]
\caption{Results of Comparison with SCOTS on examples}
\label{tab:scots}

\textbf{Legend}: $n$: \# state variables,
  $itr$ : \# iterations, Time: total computation time, $\eta$: state discretization step, $\tau$: time step. All timings are in seconds and rounded, TO: timeout ($> 10$ hours).
\begin{center}
{
\begin{tabular}{||l|l||l|l|r|r||l|r||}
\hline
\multicolumn{2}{||c||}{Problem} & \multicolumn{4}{c||}{SCOTS}
& \multicolumn{2}{c||}{CEGIS} 
\tabularnewline \hline
ID & $n$ & $\eta$ & $\tau$ & $itr$ & Time  & $\delta$ & Time
\tabularnewline \hline
Ex.~\ref{ex:obstacle-1} & 2 &
0.16$^2$ & 0.12 &
18 & 0 & 10$^{-4}$ & 3 
\tabularnewline \hline
Ex.~\ref{ex:dcdc} & 2 &
0.01$^2$ & 1.0 &
106 & 1 & 10$^{-4}$ & 39
\tabularnewline \hline
Ex.~\ref{ex:unicycle} & 3 &
0.2$^2$$\times$0.1 & 0.3 &
404 & 989 & 10$^{-4}$ & 1484
\tabularnewline \hline
Ex.~\ref{ex:basic-4D} & 4 &
0.03$\times$0.1$^3$ & 0.005 &
48 & 304 & 10$^{-5}$ & 3
\tabularnewline \hline
Ex.~\ref{ex:unicycle-4D} & 4 &
0.1$^2\times$0.05$^2$ & 0.3 &
\multicolumn{2}{|c||} {TO} & 10$^{-4}$ & 5296
\tabularnewline \hline
\end{tabular}\\
}
\end{center}
\end{table}

\section{Conclusions}
In this paper, given a switched system, we addressed controller
synthesis problems for RWS with composite barriers. Specifically, we
addressed uninitialized problems which are useful for building an
abstraction, as well as initialized problems. For each problem, we
provided sufficient conditions in terms of ``existence of a control
certificate".  Also, we demonstrated that searching for a control
certificate can be encoded into constrained problems and solving
these problems is computationally feasible.  In the future, we wish to
investigate how the initialized RWS problems can be extended to be
used along fixed-point computation based techniques as it allows more
flexible switching strategies.

\section*{Acknowledgments}
This work was funded in part by NSF under award numbers SHF 1527075 and CPS 1646556. All opinions expressed
are those of the authors and not necessarily of the NSF.

\bibliographystyle{eptcs}
\bibliography{refs}
\end{document}